\documentclass[12pt]{article}
\usepackage{bm}
\usepackage{amsmath}
\usepackage{amssymb}
\usepackage{ntheorem}
\newtheorem*{proof}{Proof}
\newtheorem{thm}{Theorem}[section]
\newtheorem{asmp}{Assumption}[section]

\newtheorem{example}{Example}[section]
\newtheorem{remark}{Remark}[section]
\usepackage{graphicx}
\usepackage[round]{natbib}
\usepackage{url} 

\title{\bf Non-iterative Gaussianization\footnote{Working paper}}
\author{Rongxiang Rui\footnote{E-mail address: raynerrui@ruc.edu.cn}   \quad  Maozai Tian\\
    \textit{School of Statistics, Renmin University of China}}
\date{}
\begin{document}
  \maketitle


\begin{abstract}
In this work, we propose a non-iterative Gaussian transformation strategy based on copula function, which doesn't require some commonly seen restrictive assumptions in the previous studies such as the elliptically symmetric distribution assumption and the linear independent component analysis assumption. Theoretical properties guarantee the proposed strategy can exactly transfer any random variable vector with a continuous multivariate distribution to a variable vector that follows a multivariate Gaussian distribution. Simulation studies also demonstrate the outperformance of such a strategy compared to some other methods like Box-Cox Gaussianization and radial Gaussianization. An application for probability density estimation for image synthesis is also shown.
\end{abstract}

\noindent%
\vfill

\section{Introduction}

As one of the most remarkable concepts of distribution transformations, Gaussianization is popularized in many different areas, such as image processing \citep{laparra2009, laparra2011} and distribution estimation \citep{chen2000, osborne2010}. 

When only considering the transformed data is supposed to follow the standard multivariate Gaussian distribution, fundamentally, three elements are included in the Gaussianization process: \textit{invertibility, marginal normality} and \textit{independence}. The \textit{invertibility} promises the application of the transformed data can be seen to some extent a similar work on the raw data, the \textit{marginal normality} guarantees the transformed data is marginally Gaussian, which is a prerequisite for the Gaussianity of the transformed distribution, and the \textit{independence} further requires the transformed data derives from a marginally independent distribution so that, combining with \textit{marginal normality}, the transformed data exactly follows a multivariate Gaussian distribution. In this regard, among previous researches, many a Gaussianization strategy is proposed based on different perspectives, which can be roughly summarized in three types: Box-Cox-based Gaussianizations, direct Gaussianizations, and iterative Gaussianizations, where the first class is parametric and the last two are nonparametric.

Box-Cox-based Gaussianizations are widely used in different research domains, e.g. biostatistics \citep{ahmad2008} and cosmologies \citep{joachimi2011, schuhmann2016}, due to their simplicity, intuitiveness, and convexity. However, the Box-Cox-based methods assume the raw data owns a Gaussian copula \citep{li2018}, such an assumption is surely too restrictive and can show some limitations in practical application.

Alternatively, some direct Gaussianization procedures are proposed for some special cases. For instance, radial Gaussianization is proposed to Gaussianize some raw data following some elliptically symmetric distributions \citep{lyu2009}, e.g. the centralized student distribution. Such a requirement is also less competitive in real data analysis.

Another type of Gaussianization approaches are about iterative (e.g., \citealp{chen2000, laparra2009, laparra2011}). Such ideas mainly include two steps in each iteration: marginal Gaussianization and linear rotation transformation. The marginal Gaussianization step makes sure the transformed data is marginally Gaussianized and the linear rotation is devoted to reconstructing a new variable vector so that the linear correlation can be alleviated by the rotation process as much as possible. In fact, due to the relationship between variables being commonly nonlinear, the linear transformation process cannot eliminate such connections. Therefore, marginal normality, which is a nonlinear transformation, is implemented. The invertibility is guaranteed by the differentiability of each iteration step \citep{laparra2011}. Despite the iteration convergence giving the courage to implement such methods, the theoretical optimal stopping iteration criterion is still unavailable.

Here, we propose a non-iterative Guassianization strategy that doesn't require a restrictive condition. Specifically, According to the basic three elements of Gaussianization mentioned above, we first implement marginal Gaussianization to ask the transformed data marginally follows the standard normal distribution, then introduce copula function to guarantee the invertibility, finally implement a re-ranking strategy to achieve marginal independence. 

In the rest of this work, we first selectively review some Gaussianization methods such as classical Box-Cox transformation, radial Gaussianization, and Gaussianization based on independent component analysis in Section \ref{sec2}. In Section \ref{sec3}, we develop our proposed approach and summarize some theoretical properties. Some comparable simulations and possible applications are shown in Section \ref{sec4} and Section \ref{sec5}, respectively. A conclusion is listed in Section \ref{sec6}.

\section{Selective review of the Gaussianization}\label{sec2}

\subsection{Box-Cox based Gaussianization}
Generally, the first-thought solution for gaussianization is to transform each random variable $X_i, i=1, \cdots, p$ marginally to a Gaussian distribution then make the transformed random vector owns a multivariate Gaussian distribution. The prerequisite underlying this strategy is a \textit{Gaussian copula assumption}, which is shown as follows.
\begin{asmp}\label{asmp1}
(Gaussian Copula Assumption; \citealp{li2018}) There exist injective functions $g_1, \cdots, g_p: \mathbb{R} \rightarrow \mathbb{R}$ such that $g_1(X_1)$, $\cdots$, $g_p(X_p)$ has a multivariate Gaussian distribution.
\end{asmp}

Under Assumption \ref{asmp1}, we have $\bm{G} = (g_1(X_1), \cdots, g_p(X_p))^\top \sim \mathcal{N}(\bm{\mu}, \bm{\Sigma})$
with some $\bm{\mu} \in \mathbb{R}$ and positive definite matrix $\bm{\Sigma} \in \mathbb{R}^{p \times p}$. Specifically, the joint density of $\bm{G}$ is 
\begin{equation*}
f_{\bm{G}} (\bm{g}) \propto \frac{1}{\det(\bm{\Sigma})^{1/2}}\exp\Big\{ -\frac{1}{2} (\bm{g} - \bm{\mu})^\top \bm{\Sigma}^{-1}(\bm{g} - \bm{\mu}) \Big\}
\end{equation*}
and the joint density function of ${X}_1, \cdots, X_p$ is 
\begin{equation}\label{eq1}
f_{\bm{X}}(\bm{x}) \propto \frac{1}{\det(\bm{\Sigma})^{1/2}}\exp\Big\{ -\frac{1}{2} (\bm{g} - \bm{\mu})^\top \bm{\Sigma}^{-1}(\bm{g} - \bm{\mu}) \Big\} \Big| \det\Big(\frac{\partial \bm{g}}{\partial \bm{x}^\top}\Big)\Big|,
\end{equation}
where $\bm{g} = (g_1(x_1), \cdots, g_p(x_p))^\top$ and $|\bm{\Sigma}|$ represents the determinant of $\bm{\Sigma}$.

Suppose we have an independent and identically distributed sample, $\bm{X}_1, \cdots, \bm{X}_n$, where $\bm{X}_l = (X_{l1}, \cdots, X_{lp})^\top, l = 1,\cdots, n$, then the joint density distribution of the sample is 
\begin{equation*}
f(\bm{x}_{l=1}^n) \propto \prod_{l=1}^n \frac{1}{\det(\bm{\Sigma})^{1/2}}\exp\Big\{ -\frac{1}{2} (\bm{g}_l - \bm{\mu})^\top \bm{\Sigma}^{-1}(\bm{g}_l - \bm{\mu}) \Big\} \Big| \det\Big(\frac{\partial \bm{g}_l}{\partial \bm{x}_l^\top}\Big)\Big|.
\end{equation*}
Therefore, one can immediately obtain the related log-likelihood as
\begin{equation*}
L(\bm{u}, \bm{\Sigma}, \bm{\theta}) \propto -\frac{n}{2}\log(|\bm{\Sigma}|) - \sum_{l=1}^n \Bigg\{(\bm{g}_l - \bm{\mu})^\top\bm{\Sigma}^{-1}(\bm{g}_l - \bm{\mu}) - \log\Big(\Big| \det\Big(\frac{\partial \bm{g}_l}{\partial \bm{x}_l^\top}\Big)\Big|\Big)\Bigg\}.
\end{equation*}

Once $g_1, \cdots, g_p$ is prescribed, by using some well-known gradient descent algorithms, the optimal estimates of $\bm{\hat{\mu}}^{(\lambda)}$, $\bm{\hat{\Sigma}}^{(\lambda)}$, and some parameters included in the injective functions could be obtained. 

For the classic Box-Cox transformation, the specific formula of function $g$ is
\begin{equation*}
g^{(BC)}(u) = \left\{
\begin{aligned}
& \frac{u ^\lambda- 1}{\lambda},& \lambda \ne 0, \\
& \log(u), &\lambda=0.
\end{aligned}
\right.
\end{equation*}

There are many other modified Box-Cox transformation methods can also be implemented and we briefly introduce some of them here. \cite{manly1976} proposed an alternative exponential transformation based on the viewpoint that the values of the samples are possibly negative, which is defined as
\begin{equation*}
g^{(ET)}(u) = \left\{
\begin{aligned}
&\frac{\exp(\lambda u) - 1}{\lambda}, &~\lambda \ne 0, \\
& u, &~ \lambda = 0.
\end{aligned}
\right.
\end{equation*}
\cite{gillard2012} proposed a generalized Box-Cox transformation with another two parameters, which is 
\begin{equation*}
g^{(GBC)} (u)= \left\{
\begin{aligned}
&\frac{((u - \alpha)/\beta)^\lambda - 1}{\lambda}, &~\lambda \ne 0, ~u >\alpha, ~\beta > 0,\\
& \log\big(\frac{u - \alpha}{\beta}\big), &~\lambda = 0, ~u > \alpha, ~\beta > 0.
\end{aligned}
\right.
\end{equation*}
\cite{schuhmann2016}, from a Bayesian perspective, proposed an extended Box-Cox transformation, named as Arcsinh-Box-Cox transformation, which is defined as 
\begin{equation*}
g^{(ABC)}(u) = \left\{
\begin{aligned}
&t^{-1}\sinh(t g^{(BC)}(u)), &~t>0, \\
&g^{(BC)}(u), &~t = 0, \\
& {\rm arcsinh} (t g^{(BC)}(u)), &~t <0,
\end{aligned}
\right.
\end{equation*}
where $t$ is deemed to remove the residual kurtosis from a model parameter distribution. For more detailed reviews see in \cite{atkinson2021}.

Box-Cox based transformation approaches could be the most intuitive ones because they directly exploit the multivariate normal distribution as the bridge to connect the raw data and the transformed data by Jacobian transformation. There are some advantages for Box-Cox and their kins. First, they provide a range of opportunities for closely calibrating a transformation to the needs of the data. For instance, when $\lambda = 1, 0, 1/2, -1$, the classic Box-Cox transformation becomes the identity transformation, the logarithmic transformation, the square root transformation, and the reciprocal transformation respectively \citep{osborne2010}. Second, some of them own the flexibility of covering a range of data by embedding different parameters e.g. the generalized Box-Cox transformation in \cite{gillard2012}. 

However, there are some non-trivial limitations for Box-Cox transformation that can't be disregarded. First, the existence of the phenomenon of the parameter explosion could destroy the possibility of getting reasonable parameter estimands \citep{joachimi2011}. Second, some special shortcomings for different transformations exist. For example, \cite{gillard2012} argued that Manly's transformation could make the transformed data own a skewness of zero but still have a histogram that is markedly asymmetric, which causes the distribution of the transformed data far away from the normality and violates the Gaussian copula assumption. As for the generalized Box-Cox transformation, it is possible to encounter practical issues when the raw data comes from a bimodal distribution, i.e. no unique values of $(\alpha, \beta, \lambda)$, \citep{gillard2012}. 

Therefore, other methods have to be developed to get reasonable Gaussianized data. Apart from the Box-Cox based transformations, some other strategies are proposed in different areas, which will be shown below.

\subsection{Direct Gaussianization}
We first introduce two direct Gaussianization approaches developed by \cite{erdogmus2006} and \cite{lyu2009}. In order to estimate the density of $\bm{X} = (X_1, \cdots, X_p)^\top$, $f(\bm{x})$,  \cite{erdogmus2006}, under a marginally Gaussianizable assumption, proposed a nonparametric Gaussianization strategy by means of a kernel density estimation technique, i.e. Parzen windowing, which applies the mininum Kullback-Leibler divergence (KLD) to choose the optimal kernel size. 
\begin{asmp}\label{asmp2}
(Marginally Gaussianizable Assumption; \citealp{erdogmus2006}) For random variable vector $\bm{X} = (X_1, \cdots, X_p)^\top$, there exists a function $\bm{h}(\cdot)$ that $\bm{\tilde{X}} = \bm{h}(\bm{X})$ follows a jointly Gaussian distribution, where $\bm{h}(\bm{X}) = (h_1({X}_1), \cdots, h_p({X}_p))^\top$.
\end{asmp}

Specifically, $h_i$'s are chosen to be marginally Gaussianized functions, i.e. $h_i(X_i) = \Phi^{-1}(F(X_i)), i = 1, \cdots, p$, where $\Phi^{-1}(\cdot)$ is the generalized inverse of the standard univariate Gaussian cumulative density function. 

In fact, one can realize that the marginally Gaussianizable assumption is fundamentally a special case of the Gaussian copula assumption and, with equation (\ref{eq1}), the joint density function of $\bm{X}$ can be refined as 
\begin{equation}\label{eq2}
f_{\bm{X}}(\bm{x}) \propto \frac{1}{\det(\bm{\Sigma})^{1/2}}\exp\Big\{ -\frac{1}{2} (\bm{h} - \bm{\mu})^\top \bm{\Sigma}^{-1}(\bm{h} - \bm{\mu}) \Big\} 
\Big|\prod_{i=1}^p\frac{f_i(x_i)}{\phi(h_i(x_i))}\Big|. 
\end{equation}

In equation (\ref{eq2}), the marginal distributions $f_i$'s then need to be estimated. \cite{erdogmus2006} introduced to use single dimensional Parzen window estimates to approximate $f_i$'s, denoted by $\hat{f}_i$'s, with KLD satisfying 
\begin{equation*}
\underset{\hat{f}_i}{\arg\min} ~{\rm KLD}(f_i || \hat{f}_i)  = \underset{\hat{f}_i}{\arg\max} E_{f_i}\big(\log(\hat{f}_i(x_i))\big)= \underset{\hat{f}_i}{\arg\max} E_{f_{i}}\big(\hat{f}_i(x_i)\big),
\end{equation*}
where 
\begin{equation*}
\hat{f}_i (u) = \frac{1}{n} \sum_{l=1}^n \mathcal{K}_\theta (u - x_{li}),  
\end{equation*}
with kernel function $\mathcal{K}_\theta (\cdot)$ and kernel size parameter $\theta$.

Alternatively, \cite{lyu2009} considered a special case that $\bm{X}$ follows an elliptically symmetric distribution and proposed a Gaussianization method called radial Gaussianization, under the following assumption. 
\begin{asmp}\label{asmp3}
(Elliptically Symmetric Assumption; \citealp{lyu2009}) The underlying distribution of the  p-dimensional raw random variable vector $\bm{X}$ is elliptically symmetric, i.e.
\begin{equation*}
k(\bm{x}) =\frac{1}{\alpha |\det((\Sigma)|^{1/2}}g(-\frac{1}{2}\bm{x}^\top\Sigma^{-1}\bm{x}),
\end{equation*}
where $\Sigma$ is a symmetric positive-definite matrix, $\alpha$ is a normalizing constant, and $g$ is a positive-valued generating function satisfying $\int_0^\infty g(-u^2/2)u^{p-1} {\rm d} u < \infty$.
\end{asmp}

For whitened variable $\bm{X}_{wht}$, the radial marginal distribution with elliptically symmetric density function, $k$, is 
\begin{equation*}
f_r(r) = \frac{r^{p-1}}{\beta}k(-r^2/2), r = ||\bm{X}_{wht}||,
\end{equation*}
where $||\bm{X}_{wht}||$ is the related Euclidean norm and $\beta$ is the normalizing constant to promise the density to integrate to one. For the standard normal distribution, the related radial marginal distribution is a chi density:
\begin{equation*}
f_{\chi}(r) = \frac{r^{p-1}}{2^{p/2 - 1}\Gamma(d/2)} \exp(- r^2/2),
\end{equation*}
where $\Gamma(\cdot)$ is the standard gamma function. With the unique transformation $t(r) = F_{\chi}^{-1}(F_r(r))$, the radial transformation is defined as
\begin{equation*}
\bm{X}_{rg} = \frac{t(r)}{r}\cdot \bm{X}_{wht}.
\end{equation*}

It is easy to see that both the marginally Gaussianizable and the elliptically symmetric are too restrictive for real cases. For the marginally Gaussianizable, if one doesn't know whether the random variable $\bm{X}$ is marginally Gaussianizable, it has to depend on the tools that can be used to identify marginally Gaussianizable components, such as principal component analysis. However, such tools themselves can also be ineffective and the related transformation could consequently be failed. For the elliptically symmetric, an example of \cite{laparra2011} directly shows that in an image context, the basic distribution is not strictly elliptically symmetric. In such cases, the radial Gaussianization is out of work. 

It's worth noting that although the limitations of the method in \cite{erdogmus2006} limit the application of the direct Gaussianization, the idea of the marginal Gaussianization is useful and has been implemented in many other studies. For example, some iterative-based Gaussianization approaches employing the marginal Gaussianization are proposed, which can alleviate the limitations mentioned above. 

\subsection{Iterative Gaussianization}
Based on a linear independent component analysis assumption (ICA), \cite{chen2000} proposed a weakly converged iterative Gaussianization procedure. 
\begin{asmp}\label{asmp4}
( Linear Independent Component Analysis Assumption; \citealp{chen2000}) There exists a linear transform $A_{p\times p}$ such that the transformed variable vector $\bm{\tilde{X}} = A\bm{X}$ owns independent components: $f(\bm{\tilde{x}}) = (f(\tilde{x}_1)\cdot \cdots \cdot f(\tilde{x}_p))^\top$, where $\bm{\tilde{X}} = (\tilde{X}_1, \cdots, \tilde{X}_p)^\top$.
\end{asmp}

Since the linear ICA assumption (Assumption \ref{asmp4}) is too restrictive for practical implementation, \cite{chen2000} developed an iterative Gaussianization procedure. That is, the previous marginally Gaussianized data is first transformed to the least (or less) dependent coordinates, then marginally Gaussianized the transformed data again, i.e. in each iteration, the raw variable vector $\bm{X}$ (the marginally Gaussianized variable vector in the previous iteration) is first linearly transformed to 
\begin{equation}\label{eq3}
\bm{\tilde{X}} = (\tilde{X}_1, \cdots, \tilde{X}_p)^\top = A\bm{X},
\end{equation} 
making the transformed data is the least dependent coordinates. Then, the transformed data is marginally Gaussianized: 
\begin{equation}\label{eq4}
\tilde{\tilde{X}}_i = \Phi^{-1}(F_i(\tilde{X}_i)), ~i = 1, \cdots, p.
\end{equation}

This iterative Gaussianization method generates an invertible and differentiable transform, making it possible to estimate the original probability density function by Jacobin transformation. However, the computational cost is the main issue that needs to be refined. In this respect, \cite{laparra2009} instead recommended using linear transform $A$ getting through linear principal component analysis (PCA). They also empirically proposed a negentropy-based stopping criterion to stop the iteration process, which is not considered in \cite{chen2000}.

\cite{laparra2011} further generalized the aforementioned iterative Gaussianization methods to some random rotations, termed rotation-based iterative Gaussianization transforms. Different to the processes in equations (\ref{eq3}) and (\ref{eq4}), their approach in each iteration first marginally Gaussianized the raw data and then linearly transformed the marginally Gaussianized data with an arbitrary rotation matrix, including ICA, PCA, and random rotations.

It's not difficult to find out that the aforementioned methods are supposed to transform the raw variable vector $\bm{X}$ to be a new random variable vector that follows a standard multivariate Gaussian distribution, i.e. $\mathcal{N}(0, I_p)$, where, $I_p$ is the $p$-dimensional identity matrix. The relationships among $X_i$'s in $\bm{X}$ are indirectly stored by each iteration process, including the linear transformation and the marginal Gaussianization process. 

However, the theoretical optimal stopping criterion is unavailable for iterative Gaussianization methods, meaning that, in practical application, the proper iteration steps are required before implementation. One can realize that Gaussianization owns three components, i.e. invertibility, marginal normality, and independence. In this respect, we develop a copula-based transformation that can not only keep invertibility but ensure the transformed data is Gaussian.

\section{Copula-based transformation}\label{sec3}
In this section, we develop a novel non-iterative Gaussianization method based on the copula methodology. Recall that according to the Sklar's theorem \citep{sklar1959}, every multivariate cumulative distribution function $F(\bm{x}) = P(X_1 \le x_1, \cdots, X_p \le x_p)$ of a random vector $\bm{X} = (X_1, \cdots, X_p)^\top$ can be expressed in terms of its marginals $F_i(x_i) = P(X_i \le x_i), i = 1, \cdots, p$ and a copula function $C$.

\begin{thm}\label{thm1}
(Sklar's Theorem, \citealp{sklar1959}) Let $F(\bm{x})$ be a p-dimensional distribution function of a random vector $\bm{X} = (X_1, \cdots, X_p)^\top$ with univariate marginals $F_i(x_i) = P(X_i \le x_i), i = 1, \cdots, p$. Then there exists a copula function $C$ such that, for every $\bm{x} = (x_1, \cdots, x_p)^\top \in \bar{\mathbb{R}}^p = [-\infty, +\infty]^p$,
\begin{equation*}
F(x_1, \cdots, x_p) = C(F_1(x_1), \cdots, F_p(x_p)),
\end{equation*}
where $\bar{\mathbb{R}}^p$ is the $p$-cartesian product of the ranges of $\bar{\mathbb{R}}$. If all $F_i(x_i)$'s are continuous, then $C$ is unique.
\end{thm}

A clearer description of Sklar's theorem can be found in \cite{nelsen2007}. Theorem \ref{thm1} pronounces that when estimating the joint distribution of $\bm{X}$, one can estimate the copula function $C$ and marginal distributions $F_i(x_i)$ separately since the relationship shown in $\bm{X}$ can be captured by the copula $C$. That is, once the copula and the marginals are well estimated, the combination of these two parts is the right estimate of $F(\bm{x})$. 

\begin{remark}\label{rem1}
One can treat the empirical copula estimate $\hat{C}$ and the empirical distribution estimates $\hat{F}_i(x_i)$'s as the reasonable estimands of $C$ and $F_i(x_i)$'s respectively, where
\begin{equation}\label{eq5}
\begin{aligned}
\hat{C} (u_1, \cdots, u_p) =& \hat{F}(\hat{F}_1^{-1} (u_1), \cdots, \hat{F}_1^{-1} (u_1)),\\
 \hat{F}(x_1, \cdots, x_p) =& \frac{1}{n} \sum_{l=1}^n \prod_{i=1}^p \bm{1}(X_{l1} \le x_1, \cdots, X_{lp} \le x_p),\\
 \hat{F}_i (x_i) =  & \sum_{l = 1}^n \bm{1}(X_{li} \le x_i), 
\end{aligned}
\end{equation}
where $\bm{1}$ is an indicator function and $\hat{F}_i^{-1}(u_i) = \inf\big\{x_i~|~\hat{F}_i(x_i) \ge u_i \big\}$, which is the empirical $u_i$-quantile of $F_i(x_i)$, $u_i \in [0, 1], i = 1, \cdots, p$. \cite{poczos2012} has given a proof that $\hat{C}$ in equation (\ref{eq5}) is a good estimate of $C$ when sample size goes larger (see more details in that work). Also, the acceptability of the empirical estimate of the marginal distribution is well studied. A simple example of the performance of the empirical copula estimation is shown in Figure \ref{fig1}, based on the R package: {\rm copula} (\url{https://cran.r-project.org/web/packages/copula/}).

Alternatively, some kernel-based estimation approaches can be also used to estimate the copula $C$ and $F_i(x_i)$, which is out of our consideration here. One can read some related studies such as \cite{poczos2012} and \cite{charpentier2007} for more details.
\end{remark}

\begin{figure}[!htbp]
\includegraphics[width=\textwidth]{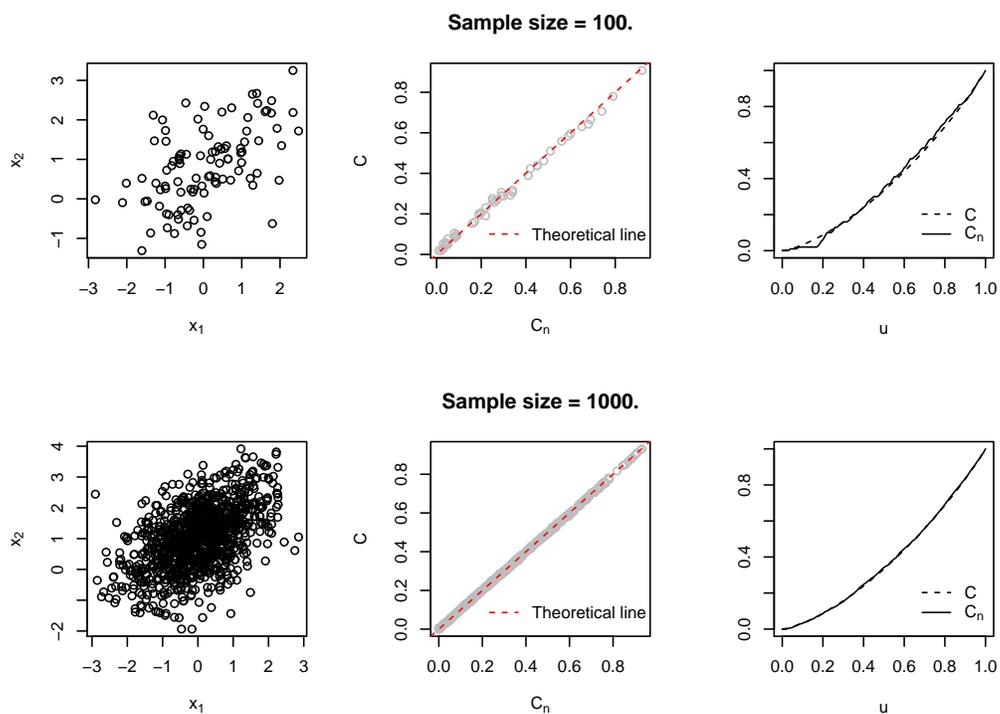}
\caption{True versus empirical diagonal of a Gaussian copula. From left to right, the first column is about the sample from $\mathcal{N}(\mu, \Sigma), \mu = (0, 0)^\top, (\Sigma)_{ab} = 2^{-|a-b|}, a, b = 1,  2$, the second column is the related pointwise comparison, and the third column is the diagonalized empirical copula (solid line) and the diagonalized underlying copula (dashed line). The top row is about the sample size of 100 and the bottom row is about the sample size of 1000.}
\label{fig1}
\end{figure}

Based on Theorem \ref{thm1}, we develop a non-iterative Gaussianization transformation method. Briefly speaking, our method mainly include three steps: 1) marginally Gaussianize the raw data, 2) estimate the copula $C$ of the data from 1), and 3) make the data from 1) independent under keeping the invariance of the marginal distribution of the data from 1). That is, first, denoted by $X_i^{*}, i = 1, \cdots, p$ the marginally Gaussianized variable, 
\begin{equation} 
X_i^{*} = \Phi^{-1}(F_i(X_i)), i = 1, \cdots, p,
\end{equation}\label{eq6}
which promises the marginal distributions are normally distributed with mean 0 and variance 1, i.e. $X_i^{*} \sim \mathcal{N}(0, 1), i = 1, \cdots, p$. Second, based on $X_i^{*}$, we obtain the copula of the joint cumulative distribution, indicated be $H$, under Theorem \ref{thm1},
\begin{equation*}
H(x^{*}_1, \cdots, x^{*}_p) = C(\Phi(x^{*}_1), \cdots, \Phi_p(x^{*}_1)).
\end{equation*}
Since the marginal Gaussianization can't omit the relationship included in $\bm{X}$ unless Assumption \ref{asmp4} is strictly met (In this case, whitening the raw data can omit linear dependence), Thus, we still need to do some treatments making the treated data is independent and the marginal distribution is still hold. Here, a reranking method is introduced to achieve our goal. 

We start the final step through sample perspective. Suppose $\bm{X}_1$,$\cdots$, $\bm{X}_{n}$ are $n$ duplications of $\bm{X}$, where $\bm{X}_l = (X_{l1}, \cdots, x_{lp})^\top$, $l = 1, \cdots, n$. Let $\bm{X}^{i}  = (X_{1i}, \cdots, X_{ni})^\top$, $i = 1, \cdots, p$ denote the $n$ observations of $X_i$ and $\bm{X}^{*i} = (X^{*}_{1i}, \cdots, X^{*}_{ni})$ represent the $n$ marginally Gaussianized observations of $X_i$. We get $n$ artificial observations by re-ranking to generate pseudo-random variable vector (denoted by $\bm{\tilde{X}}^{*} = ({\tilde{X}}^{*1}, \cdots, {\tilde{X}}^{*p})^\top$), each denoted by $\bm{\tilde{X}}^{*i} = (X^{*}_{1i}, \cdots, X^{*}_{ni})^\top$,  where
\begin{equation}\label{eq7}
 \bm{\tilde{X}}^{*i} = (X^{*}_{i+\delta_i i}, \cdots, X^{*}_{ni}, X^{*}_{1i}, \dots, X^{*}_{i + \delta_i-1i})^\top, ~i = 1, \cdots, p,
\end{equation}
$\delta_1 =0, \delta_2 = \cdots = \delta_p \in \{ 0, \cdots, \lfloor n/p \rfloor \}$, where $\lfloor n/p \rfloor$ denotes the largest integer not greater than $n/p$. 

The invariant of the marginal distributions of $F_i(x_i)$ estimated by the pseudo-random variable ${\tilde{X}}^{*i}$ is clear, since, primarily, the pseudo-random variables generated by the re-ranking strategy shown in Equation (\ref{eq7}) concatenates the different observations of different variables so that the generated pseudo-random variables are deemed to be independent and still keep the marginals invariant. Since the sample size of the pseudo-random variable through the re-ranking strategy is fully equal to that of the observations of the raw random variable, the generated observations, are the marginally ranked ones. Therefore, it is impossible to encounter some cases in that some of the generated observations are exactly from a duplication of the raw variable vector, which means the generated pseudo-random variables are independent. A simple example is shown below.

\begin{example}\label{exam1}
For simplicity, we provide an example of 2-dimensional case. It is well-known that body height ($X^{h}$) and weight ($X^{w}$) show significant correlation. Suppose that we have a $n$-size sample which is independent and identically distributed, say $(X_1^{h}, X_1^{w})$, $\cdots$,  $(X_n^{h}, X_n^{w}) \sim F$. Apparently, $X_l^{h}$ and $X_m^{w}, l \ne m$ is independent since they come from two different observations. Based  on Equation (\ref{eq7}), the transformed data is as $(X_1^{h}, X_2^{w})$, $\cdots$, $(X_{n-1}^{h}, X_n^{w}), (X_n^{h}, X_1^{w})$,  promising the generated pseudo-random variable is exactly independent. This can be guaranteed from the fact that the observations $(X_l^{h}, X_l^{w})$ and $(X_m^{h}, X_m^{w}), l \ne m$ are independent. The next theorem summarizes this property.
\end{example}

\begin{thm}\label{thm2}
For 2-dimensional random variable vector $\bm{X} = (X_1, X_2)^\top$, denote by $\bm{X}^{*} = (X^{*}_1, X^{*}_2)^\top$ the marginally Gaussianized random variable vector. Suppose we have $n$ observations of $\bm{X}$, and denoted by $(\tilde{X}^{*}_{11}, \tilde{X}^{*}_{12})^\top, \cdots, (\tilde{X}^{*}_{n1}, \tilde{X}^{*}_{n2})^\top$ the marginally Gaussianized observations. Then, the generated pseudo-random variable vector $\bm{\tilde{X}}^{*} = (\tilde{X}^{*}_1, \tilde{X}^{*}_2)^\top$ by means of Equation (\ref{eq7}) satisfies\\
1) the marginal distributions of $\bm{\tilde{X}}^{*}$ are equivalent to that of $\bm{X}^{*}$;\\
2) $\tilde{X}^{*}_1$ and $\tilde{X}^{*}_2$ are independent;\\
3) hence, $\bm{\tilde{X}}^{*}$ follows the standard 2-dimensional Gaussian distribution, i.e. $\bm{\tilde{X}}^{*} \sim \mathcal{N}(0, I_2)$.
\end{thm}
\begin{proof}
1) Recall that, from Equation (\ref{eq7}), the observations of $\tilde{X}^{*}_1$ and $\tilde{X}^{*}_2$ are only the observations with changed permutations respectively, which does not change the related population distribution. Hence, the marginal distributions of $\bm{\tilde{X}}^{*} = (\tilde{X}^{*}_1, \tilde{X}^{*}_2)^\top$ are equivalent to that of $\bm{X}^{*}$. Since the marginal distributions of $\bm{X}^{*}$ are the standard normal distribution, i.e. $\mathcal{N}(0, 1)$, the marginal distributions of the transformed data are also the standard normal distribution.

2) We first prove that $(\tilde{X}^{*}_{11}, \tilde{X}^{*}_{12})^\top, \cdots, (\tilde{X}^{*}_{n1}, \tilde{X}^{*}_{n2})^\top$ follow an identical distribution. For the sake of simplicity, we only consider $\bm{\tilde{X}}_1^{*} = (\tilde{X}^{*}_{11}, \tilde{X}^{*}_{12})^\top$, $\bm{\tilde{X}}_2^{*} = (\tilde{X}^{*}_{21}, \tilde{X}^{*}_{22})^\top$ and fix $\delta_1 = \delta_2 = 0$. In fact, the joint distribution of $\bm{X}_1^{*}$ and $\bm{\tilde{X}}_2^{*}$ satisfies
\begin{equation*}
f(u^{*}_{11}, u^{*}_{12}, u^{*}_{21}, u^{*}_{22}) = \int f(x^{*}_{11}, x^{*}_{12}, x^{*}_{21}, x^{*}_{22}, x^{*}_{31}, x^{*}_{32}) {\rm d} x^{*}_{12} {\rm d} x^{*}_{31}.
\end{equation*}
\begin{equation}
\begin{aligned}
f(u^{*}_{11}, u^{*}_{12}) = & \int \Bigg( \int f(x^{*}_{11}, x^{*}_{12}, x^{*}_{21}, x^{*}_{22}, x^{*}_{31}, x^{*}_{32}) {\rm d} x^{*}_{12} {\rm d} x^{*}_{31} \Bigg) {\rm d} x^{*}_{21} {\rm d} x^{*}_{32}, \\
=& \int \Bigg( \int f(x^{*}_{11}, x^{*}_{12})f(x^{*}_{21}, x^{*}_{22})f(x^{*}_{31}, x^{*}_{32}) {\rm d} x^{*}_{12} {\rm d} x^{*}_{31} \Bigg) {\rm d} x^{*}_{21} {\rm d} x^{*}_{32},\\
\equiv &  \phi(u^{*}_{11})\phi(u^{*}_{12}), ~u^{*}_{11}, u^{*}_{12} \in R,
\end{aligned}
\end{equation}
where $\phi$ is the standard normal probability density function. Similarly, 
\begin{equation*}
f(u^{*}_{21}, u^{*}_{22}) \equiv \phi(u^{*}_{21})\phi(u^{*}_{22}), ~u^{*}_{21}, u^{*}_{22} \in R.
\end{equation*}
Therefore, one can realize that $\bm{\tilde{X}}_1^{*}$ and $\bm{\tilde{X}}_2^{*}$ follow an identical distribution and, furthermore, see that $\tilde{X}^{*}_1$ and $\tilde{X}^{*}_2$ are independent. 

3) The related proof is apparent. From the proof of 2), one can easily recognize that the joint distribution function of $\tilde{X}^{*}_1$ and $\tilde{X}^{*}_2$ is merely the product of two standard normality distribution function, which directly implies the result.
\end{proof}

Based on Theorem \ref{thm2}, we can immediately generalize it to $p$-dimensional cases shown below.
\begin{thm}\label{thm3}
For 2-dimensional random variable $\bm{X} = (X_1, \cdots, X_p)^\top$, denote by $\bm{X}^{*} = (X^{*}_1, \cdots, X^{*}_p)^\top$ the marginally Gaussianized random variable. Suppose we have $n$ observations of $\bm{X}$, and denoted by $(\tilde{X}^{*}_{11}, \cdots, \tilde{X}^{*}_{1p})^\top, \cdots, (\tilde{X}^{*}_{n1}, \cdots, \tilde{X}^{*}_{np})^\top$ the marginally Gaussianized observations. Then, the generated pseudo-random variable $\bm{\tilde{X}}^{*} = (\tilde{X}^{*}_1, \cdots, \tilde{X}^{*}_p)^\top$ shown in Equation (\ref{eq7}) satisfies\\
1) the marginal distributions of $\bm{\tilde{X}}^{*}$ are equivalent to that of $\bm{X}^{*}$;\\
2) $\tilde{X}^{*}_1, \cdots, \tilde{X}^{*}_p$ are independent;\\
3) hence, $\bm{\tilde{X}}^{*}$ follows the standard p-dimensional Gaussian distribution, i.e. $\bm{\tilde{X}}^{*} \sim \mathcal{N}(0, I_p)$.
\end{thm}

\begin{remark}
The proof of Theorem \ref{thm3} is very similar to that of Theorem \ref{thm2} by given $\delta_1 = \cdots \delta_p = 0$ and we omit it here. Theorem \ref{thm3} guarantees the transformed data follows a multivariate Gaussian distribution, i.e. $\mathcal{N}(0, I_p)$.
\end{remark}

Here, we integrate the above procedures in Algorithm 1 based on the sample perspective. The use of the copula for our Gaussianization method is a key step, which makes it possible to transform the raw data into an independent one. This is because the copula can completely capture the relationships among random variables. Once the copula is determined, it is possible to make the marginals independent.
\begin{table}[htbp]
\resizebox{0.8\width}{!}{
\begin{tabular}{ll}
\multicolumn{2}{l}{\textbf{Algorithm 1}: Copula-based Gaussianization procedure.} \\
\hline
Step 0 (Initialization): &Given $n$ independent and identical observations $\bm{X}_1, \cdots, \bm{X}_n$, \\ &equivalently $\bm{X}^{*1}, \cdots, \bm{X}^{*p}$, $\delta_1, \cdots,\delta_p$ (see detail explanation\\
& in Remark \ref{rem3}).\\
Step 1 (Empirical copula estimation): & Estimate empirical copula $\hat{C}$ shown in equation (\ref{eq5}).\\
Step 2 (Marginal Gaussianization):& Marginally Gaussianize the observations using equation (\ref{eq6}),\\
& denoted by $\bm{X}^{*i} = (X^{*}_{1i}, \cdots, X^{*}_{ni}), i=1, \cdots, p$,\\
&\multicolumn{1}{c}{$X^{*}_{li} = \Phi^{-1}(\hat{F}_i({X}_{li})), l = 1, \cdots, n,$}\\
& where $\hat{F}_i$ is the related empirical distribution shown in \\
& Remark \ref{rem1}. \\
Step 3 (Re-ranking): &Generate observations of the pseudo-random variable under \\
&equation (\ref{eq7}).\\
\hline
\end{tabular}}
\end{table}

\begin{remark}\label{rem3}
From Theorem \ref{thm2} and \ref{thm3}, one can see that by setting $\delta_1 = \cdots \delta_p = 0$, the Gaussianity can be satisfied. In fact, the value of $\delta_1, \cdots, \delta_p$ can be given differently and only the configuration that different dimensions correspond to different observations is required. In this respect, the total number of possible combinations of $\delta_1, \cdots, \delta_p$ is 
\begin{equation*}
n\cdot(n-1)\cdot \cdots \cdot(n-p+1).
\end{equation*}
 To simplify the related issue, we only consider a special case. Furthermore, to determine the value of $\delta_1, \cdots, \delta_p$, the Gaussianity test can be used to select the optimal ones.
\end{remark}

Interestingly, one can keep in mind that the goal of the implementation of the copula is to make the proposed Gaussianization transformation is invertible (see an example in Figure \ref{fig2}), which makes the application of the transformed data is connected to the raw data. Therefore, if one is only interested to use the transformed data, Step 1 in Algorithm 1 can be disregarded.
\begin{figure}[!htbp]
\centering
\includegraphics[width=\textwidth]{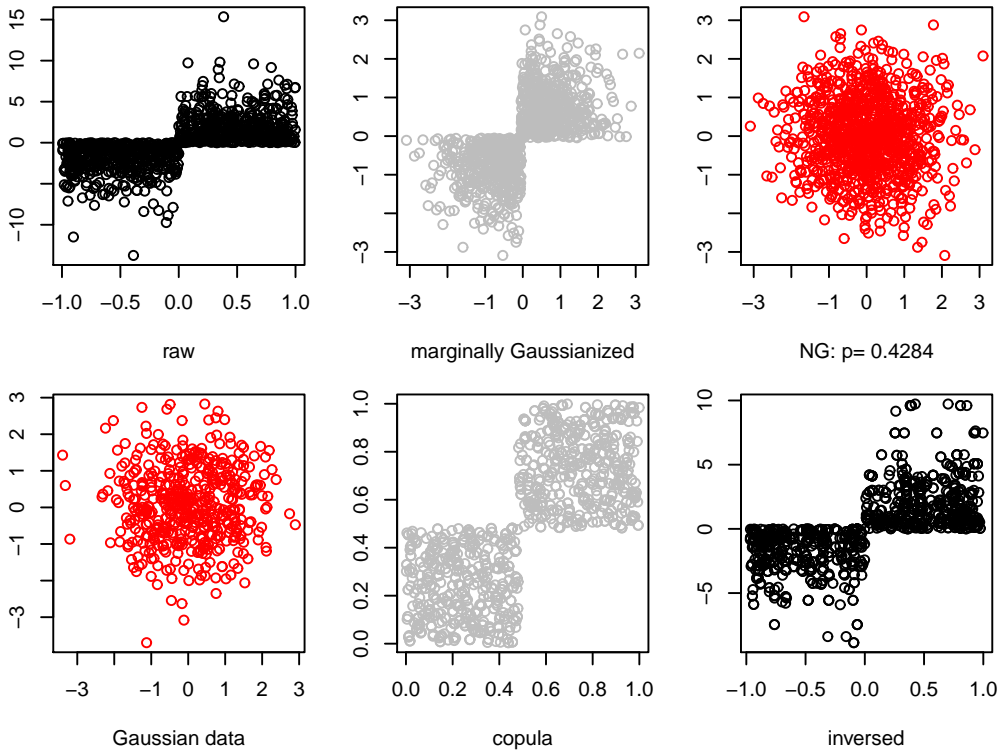}
\caption{A toy example of the invertibility of the non-iterative Gaussianization (NG, $\delta_1=\delta_2 = 0$). The top raw is about the Gaussianization process and the bottom raw is about the inverse Gaussianization process. The raw data for forward Gaussianization is generated as follows: $X_1 \sim \mathcal{U}(-1,1), Z \sim \chi^2(2), X_2 = |Z|{\rm sign}(X_1)$, and the Gaussian raw data in the inverse process is generated as: $(Z_1, Z_2)^\top \sim \mathcal{N}(0, I_2)$. The $p$-value in NG is about the related Shapiro-Wilk multivariate normality test for transformed data by means of the non-iterative Gaussianization. The value of $0.4284 ~(\gg 0.01)$ implies the transformed data is significantly close to a multivariate Gaussian distribution. }
\label{fig2}
\end{figure}
It is also worthy of noting that, first, although, in daily studies, the case could encounter that the marginal distributions of some of $X_i, i=1, \cdots, p$ are discontinuous over some points, which doesn't meet the unique condition in Theorem \ref{thm1}, we disregard this case and assume that the marginals $F_i(x_i)$'s are continuous, as all the aforementioned transformations did.

\section{Simulation study}\label{sec4}
Here, we demonstrate the capability of our non-iterative Gaussianization (NG) comparing to some candidates---Box-Cox Gaussianization (BCG; e.g. \citealp{li2018}), radial Gaussianization (RG; \citealp{lyu2009}), and rotation-based iterative Gaussianization (RBIG; \citealp{laparra2011}) with the principle component analysis rotation. To evaluate the performance of the transformation methods, Shapiro-Wilk multivariate normality test is implemented to determine, statistically, whether the transformed data is significantly close to the standard multivariate Gaussian distribution. Also, KLD,
\begin{equation}
{\rm KLD}(\phi || \hat{\phi}) = \int \phi(\bm{x}) \log\Big( \frac{\phi(\bm{x})}{\hat{\phi}(\bm{x})}\Big) {\rm d} \bm{x} = E_{\phi} (\log(\phi(\bm{x}))) - E_{\phi} (\log(\hat{\phi}(\bm{x}))),
\end{equation}
is calculated to quantify the performance of the related transformations. The smaller value of the {\rm KLD} implies the better performance of the related transformation method. In numerical process, uniformly random selection of the related support is applied so that the computation cost can be alleviated. Specifically, the number of selected points is set to be 1000, and 50 replications are carried out to report the mean and related standard deviation (sd) of KLD. 

To comprehensively compare the selected methods, 4 cases are considered and each case is corresponding to two different settings: $p=2$ and $p=4$. For BC and RBIG, iteration stop criteria are needed. We ideally set maximum iteration steps as 30 and 50, separately. Besides, the parameter range for BC is set to be -2 to 2.

\textbf{Case 1}: The underlying distribution of the raw data is an elliptically symmetric distribution with a Gaussian copula. Obviously, the setting in this case meets both the assumptions shown in Assumption \ref{asmp1} and \ref{asmp3}, which means that the true distribution is fairly friendly to BCG and RG. Concretely, without loss of generality, the cumulative distribution of the random variable vector $\bm{X}$ is 
\begin{equation*}
F(\bm{x}) = C_{Gau}(F_1(x_1), \cdots, F_p(x_p)), 
\end{equation*}
where $C_{Gau}(\cdot)$ is a Gaussian copula with correlation $\rho_{ab} = 2^{-|a - b|},  a, b =1, \cdots, p $, $X_i \sim t(6), i = 1, \cdots, p$, $F_i(x_i)$ is the corresponding cumulative marginal distribution of student distribution with degree of freedom 6.

\textbf{Case 2}: The underlying distribution of the raw data is a non-elliptically symmetric distribution with a Gaussian copula. In this case, only Gaussian copula assumption (i.e., Assumption \ref{asmp1}) is met. Thus, BCG should theoretically perform better than other three methods. Specifically, the cumulative distribution of the random variable $\bm{X}$ is 
\begin{equation*}
F(\bm{x}) = C_{Gau}(F_1(x_1), \cdots, F_p(x_p)).
\end{equation*}
Similar to \textbf{Case 1}, $C_{Gau}$ is a Gaussian copula with $\rho_{ab} = 2^{-|a - b|},  a, b = 1, \cdots, p $ whereas $X_i \sim \exp(1), i = 1, \cdots, p$.

\textbf{Case 3}: The underlying distribution of the raw data is an elliptically symmetric distribution with a non-Gaussian copula. In this setting, in contrast, the only elliptically symmetric condition is warranted, which implies the RG is the theoretically best one. To meet the requirement, here, we use instead student-$t$ copula and the marginals with normal distribution, i.e.
\begin{equation*}
F(\bm{x}) = C_{t}(F_1(x_1), \cdots, F_p(x_p)), 
\end{equation*}
where $X_i \sim \mathcal{N}(0, 1), i = 1, \cdots, p$ and the correlation is applied as the same as that shown in \textbf{Case 1}.

\textbf{Case 4}: The underlying distribution of the raw data is a non-elliptically symmetric distribution with a non-Gaussian copula. this configuration clearly conflicts to both Assumption \ref{asmp1} and \ref{asmp3} and the performance should theoretically be malfunctioned. Particularly, for the random variable vector $\bm{X}$, the related cumulative distribution is given by 
\begin{equation*}
F(\bm{x}) = C_{Clay}(F_1(x_1), \cdots, F_p(x_p)), 
\end{equation*}
where $C_{Clay}(\cdot)$ denotes the Clayton copula with parameter $\theta = 3$ and $X_i \sim \exp(1), i = 1, \cdots, p$.

\begin{table}[!htbp]
\centering
\caption{Comparison of NG, RBIG, BCG, and RG in different cases with dimension $p = 2$. In each case, the sample size is equal to 1000, 1500, and 2000, respectively. For KLD, the reported includes the mean of 50 duplications and the related standard deviation (sd) in the parenthesis.}
\label{tab2}
\resizebox{0.9\textwidth}{!}{
\begin{tabular}{cc|cccc}
\hline
\hline
&&NG &RBIG & BCG& RG \\
\hline
&&\multicolumn{4}{c}{Sample size $n = 1000$}\\
\hline
Case 1: &$p$-value & 0.2046 & 0.9389 & 0.0000 & 0.1143\\
&KLD & 0.0313  (0.0026) & 0.0258 ( 0.0024) & 0.0155 ( 0.0022) &  0.1034 ( 0.0153)\\
Case 2: &$p$-value & 0.1347 & 0.7396 & 0.0454 & 0.0016\\
&KLD & 0.0290  (0.0028) & 0.0262  ( 0.0031) & 0.0416  ( 0.0028 ) & 0.1647 (0.0470)\\
Case 3: &$p$-value & 0.9687 & 0.9996 & 0.0000 & 0.0000 \\
&KLD &0.0283 (0.0023 ) &  0.0286 (0.0024) &  0.0304 ( 0.0029) &  0.0351 ( 0.0033)\\
Case 4: &$p$-value & 0.4748 & 0.0708 & 0.0000 & 0.0000\\
&KLD & 0.0257 (0.0027) & 0.0276  ( 0.0022 ) & 0.0629  (0.0046) &0.1042 ( 0.0118)\\
\hline
&&\multicolumn{4}{c}{Sample size $n = 1500$}\\
\hline
Case 1: &$p$-value &0.3372 & 0.9641 & 0.0000 & 0.1856\\
&KLD & 0.0272 (0.0024 ) & 0.0276  (0.0020 ) & 0.0206  (0.0025) & 0.0454 ( 0.0067)\\
Case 2: &$p$-value & 0.0269 & 0.4336 & 0.2398 & 0.0000\\
&KLD & 0.0302  (0.0025) & 0.0241 ( 0.0027) &  0.0415  ( 0.0144) &0.1775 ( 0.0833)\\
Case 3: &$p$-value& 0.5275 & 0.9992 & 0.0000 & 0.0000\\
&KLD& 0.0280  (0.0021 ) & 0.0251 (0.0025 ) & 0.0237   (0.0026 ) & 0.1115 (0.0171)\\
Case 4: &$p$-value & 0.6694 & 0.2834 & 0.0000 & 0.0000\\
&KLD & 0.0264  (0.0025 ) &0.0282 (0.0028) &  0.0624 ( 0.0200 ) &  0.5408 (0.1279)\\
\hline
&&\multicolumn{4}{c}{Sample size $n = 2000$}\\
\hline
Case 1: &$p$-value &0.4078 & 0.9159 & 0.0000 & 0.0000\\
&KLD & 0.0276  (0.0023) &0.0291 ( 0.0023 ) &  0.0413  (0.0031 ) & 0.6460 (0.0819)\\
Case 2: &$p$-value &0.0601 & 0.9738 & 0.8426 & 0.0000\\
&KLD & 0.0302 (0.0026) & 0.0273 ( 0.0021 ) & 0.0411 (0.0197) &  1.7214 ( 0.2693)\\
Case 3: &$p$-value & 0.6657 & 0.9989 & 0.0000 & 0.0000\\
&KLD & 0.0303  (0.0024) & 0.0259  ( 0.0024) & 0.0281  ( 0.0027) & 0.0379 ( 0.0039)\\
Case 4: &$p$-value & 0.4753 & 0.0824 & 0.0000 & 0.0000\\
&KLD &  0.0295  (0.0026 ) & 0.0287 (0.0029) &  0.0895  (  0.0593 ) & 0.1852 (0.0384)\\
\hline
\hline
\end{tabular}}
\end{table}

\begin{table}[!htbp]
\centering
\caption{Comparison of NG, RBIG, BCG, and RG in different cases with dimension $p = 4$. In each case, the sample size is equal to 1000, 1500, and 2000, respectively. For KLD, the reported includes the mean of 50 duplications and the related standard deviation (sd) in the parenthesis.}
\label{tab3}
\resizebox{0.9\textwidth}{!}{
\begin{tabular}{cc|cccc}
\hline
\hline
&&NG &RBIG & BCG& RG \\
\hline
&&\multicolumn{4}{c}{Sample size $n = 1000$}\\
\hline
Case 1: &$p$-value &0.5010 & 0.8723 & 0.0000 & 0.0205\\
&KLD & 0.1912 (0.0262) & 0.1736 (0.0055) & 0.1913 (0.0253) & 0.1895 (0.0052)\\
Case 2: &$p$-value & 0.1254 & 0.5907 & 0.4135 & 0.0000\\
&KLD & 0.1939 (0.0342) & 0.1874 (0.0308) & 0.1961 (0.0354) & 0.4058 (0.1229)\\
Case 3: &$p$-value & 0.3640 & 0.9829 & 0.0000 & 0.0000 \\
&KLD &0.1819 (0.0068) & 0.175 (0.0155) & 0.1795 (0.0152) & 0.1785 (0.0049)\\
Case 4: &$p$-value & 0.3012 & 0.4711 & 0.0000 & 0.0000\\
&KLD & 0.1898 (0.0247) & 0.1748 (0.0161) & 0.2578 (0.0809) & 0.3694 (0.0952)\\
\hline
&&\multicolumn{4}{c}{Sample size $n = 1500$}\\
\hline
Case 1: &$p$-value & 0.6292 & 0.7564 & 0.0000 & 0.2235\\
&KLD & 0.2495 (0.0419) & 0.2268 (0.0260) & 0.2486 (0.0364) & 0.2590 (0.0473)\\
Case 2: &$p$-value & 0.4828 & 0.9561 & 0.4469 & 0.0000\\
&KLD & 0.2450 (0.0337) & 0.2442 (0.0502) & 0.2668 (0.0556) & 0.9969 (0.2462)\\
Case 3: &$p$-value& 0.3879 & 0.7515 & 0.0000 & 0.0000\\
&KLD& 0.2609 (0.0588) & 0.2453 (0.0441) & 0.2266 (0.0220) & 0.2259 (0.0152)\\
Case 4: &$p$-value & 0.2641 & 0.8126 & 0.0000 & 0.0000\\
&KLD & 0.2391 (0.0364) & 0.2403 (0.0394) & 0.5746 (0.1458) & 0.6773 (0.1935)\\
\hline
&&\multicolumn{4}{c}{Sample size $n = 2000$}\\
\hline
Case 1: &$p$-value & 0.0724 & 0.6133 & 0.0000 & 0.0000\\
&KLD & 0.3114 (0.0655) & 0.2726 (0.0343) & 0.3047 (0.0682) & 1.2669 (0.2437)\\
Case 2: &$p$-value & 0.0468 & 0.2681 & 0.2854 & 0.0000\\
&KLD & 0.3315 (0.0703) & 0.3466 (0.0814) & 0.3735 (0.1056) & 1.6215 (0.3597)\\
Case 3: &$p$-value & 0.0777 & 0.7985 & 0.0000 & 0.0000\\
&KLD & 0.3261 (0.0752) & 0.2873 (0.0451) & 0.2796 (0.0389) & 0.2784 (0.0244)\\
Case 4: &$p$-value & 0.4691 & 0.8401 & 0.0000 & 0.0000\\
&KLD & 0.2937 (0.0484) & 0.3354 (0.0883) & 0.7867 (0.2062) & 1.1878 (0.2500)\\
\hline
\hline
\end{tabular}}
\end{table}

The summarized materials are shown in Tabel \ref{tab2} and \ref{tab3}. As we mentioned in Theorem \ref{thm2} and \ref{thm3}, our developed Gaussianization method, NG, is theoretically Gaussian distributed. This is also can be empirically proved from the simulation study, i.e. Shapiro-Wilk multivariate normality test results. One can see that the normality of both NG and RBIG is remarkably higher than the other two candidates. Under threshold $\alpha = 0.01$, all outcomes of BG and RBIG are statistically Gaussianized and for some cases, RBIG performs better than NG whereas NG is the best one for other cases. 
\begin{figure}[htbp]
\centering
\includegraphics[width = \textwidth]{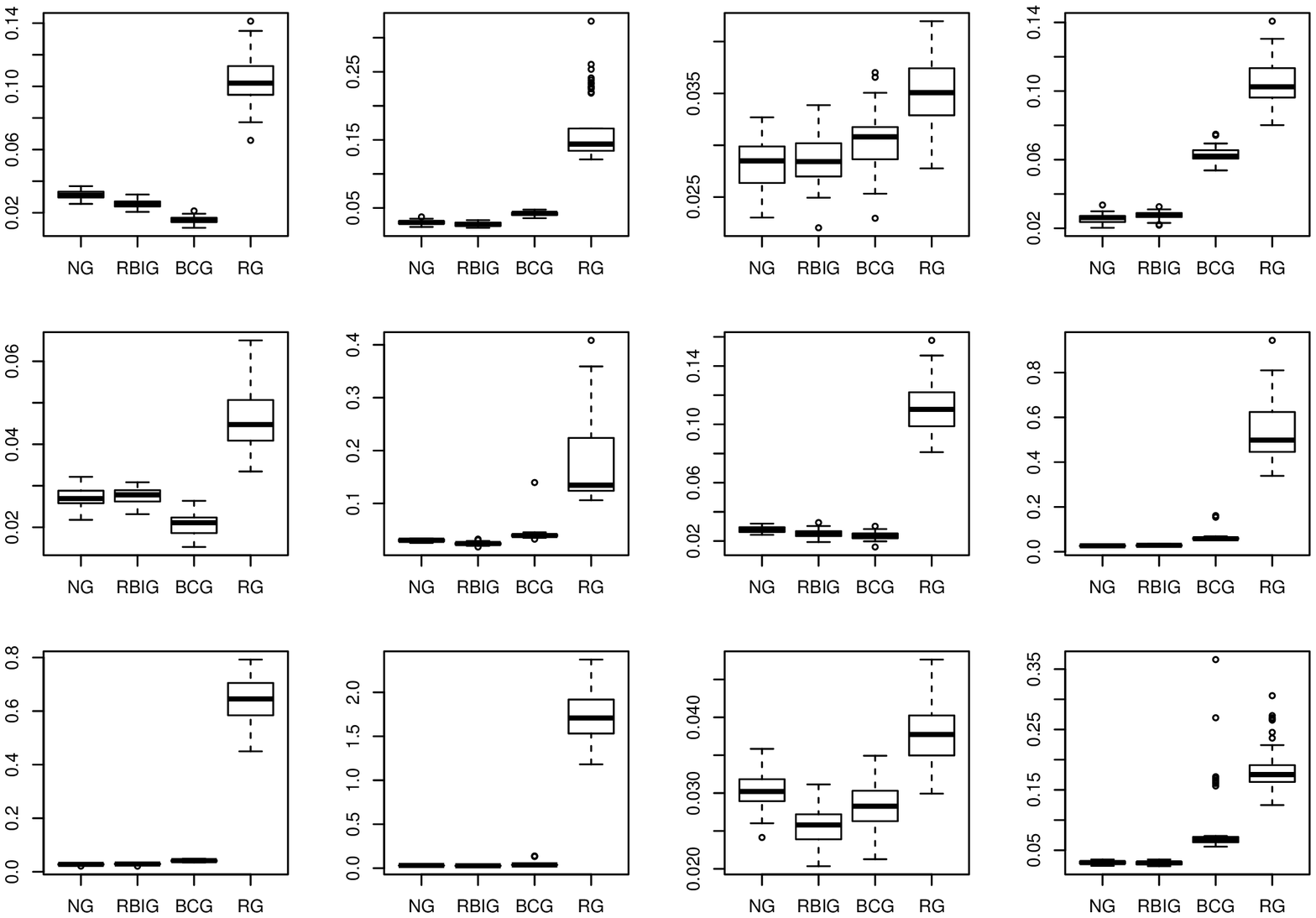}
\caption{KLD with different cases and different sample size under $p=2$. The first row is about the sample size $n = 1000$, the second row is about the sample size $n=1500$, and the third row is about the sample size $n=2000$. From left to right, each column indicates \textbf{Case 1} \textbf{Case 2}, \textbf{Case 3}, and  \textbf{Case 4} respectively.}
\label{fig3}
\end{figure}

\begin{figure}[htbp]
\centering
\includegraphics[width = \textwidth]{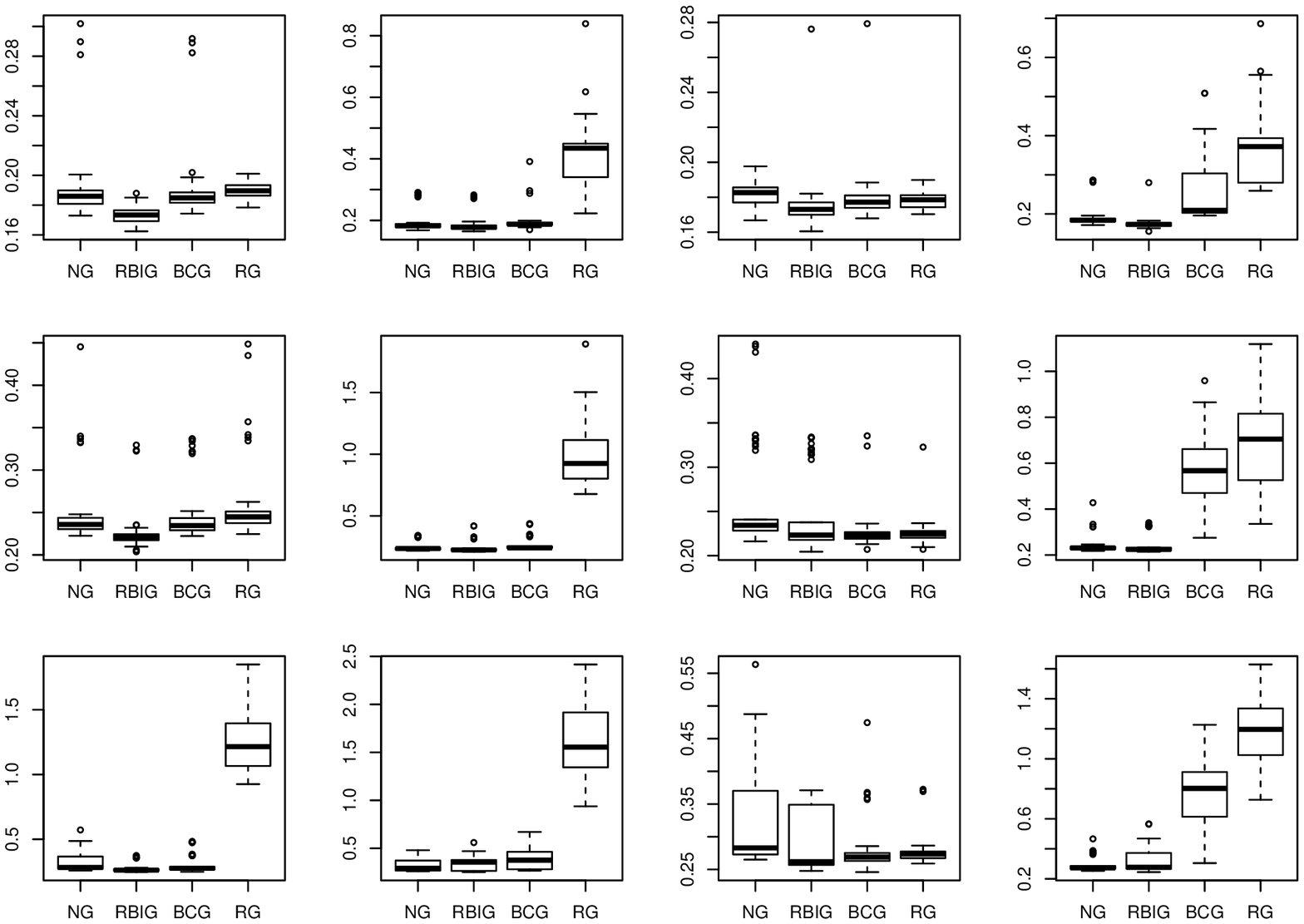}
\caption{KLD with different cases and different sample size under $p=4$. The first row is about the sample size $n = 1000$, the second row is about the sample size $n=1500$, and the third row is about the sample size $n=2000$. From left to right, each column represents  \textbf{Case 1} \textbf{Case 2}, \textbf{Case 3}, and  \textbf{Case 4} respectively.}
\label{fig4}
\end{figure}

It's worth of mentioning that even if Assumption \ref{asmp1} is met in both \textbf{Case 1} and \textbf{Case 2}, the BCG is significantly far away from Gaussianity in \textbf{Case 1}. The possible reasons can be 1) the need for the choice of a much larger range of the parameter, 2) the necessity of giving a larger step size, and 3) the possibility of parameter explosion that damages the convergence. Meanwhile, although \textbf{Case 3} is intended to RG, the performance is discouraged. The KLD's for almost all cases are larger than other methods ( see Figure \ref{fig3}, \ref{fig4}), which could mean that RG is too sensitive for the underlying distribution, and for some relatively complicated structures of the distributions, the performance of RG is less competitive.

\section{Image synthesis}\label{sec5}
In the past few decades, a lot of image synthesis methods based on Generative Adversarial Networks are proposed and some of them can be found in a selective review by \cite{wang2020}. Also, some approaches based on the iterative Gaussianization are proposed, see an example in \cite{laparra2009}. Here, we focus on the proposed NG method. Recall that NG is invertible due to the application of the copula function (see an illustration in Figure \ref{fig2}). Here, we apply NG to describe the probability density function of the Cropped Yale Face Data B with 2414 faces, which can be downloaded at \url{http://vision.ucsd.edu/~leekc/ExtYaleDatabase/ExtYaleB.html}. Images pre-cropped to 30$\times$30 pixels are used to estimate the copula function. Then, the inverse of the transform is carried out by samples generated from the standard multivariate Gaussian distribution $\mathcal{N}(0, I_{900})$. Figure \ref{fig10} shows 8 real faces at the top row and the other 32 synthesized faces. The synthesized images are a realistic representation of the learned probability density function.

\section{Conclusion}\label{sec6}
In this work, we propose a copula-based non-iterative Gaussianization strategy that can be used for multivariate probability function estimation. The method promises the transformed data exactly derives from the standard multivariate Gaussian distribution. By comparing with some other Gaussianization methods, NG shows competitive strength for different settings and we show the ability of NG for probability density function estimation by a simple application in image synthesis. 

One limitation of this method is clear that the transformed observations guarantee the identical distribution condition but do not satisfy the theoretically independent requirement, even if in the empirical studies, such a disadvantage doesn't show an apparent impact on the result. Another key point that one should keep in mind is that the variables are assumed to follow a continuous multivariate distribution. Such a prerequisite implies that 1) for a case that variables include some nonrandom variables, our proposed method should not be considered, 2) since continuity for the distribution is prescribed, our strategy is also should not be implemented for cases that contain discrete variables.

\begin{figure}[!htbp]
\centering
\includegraphics[width = \textwidth]{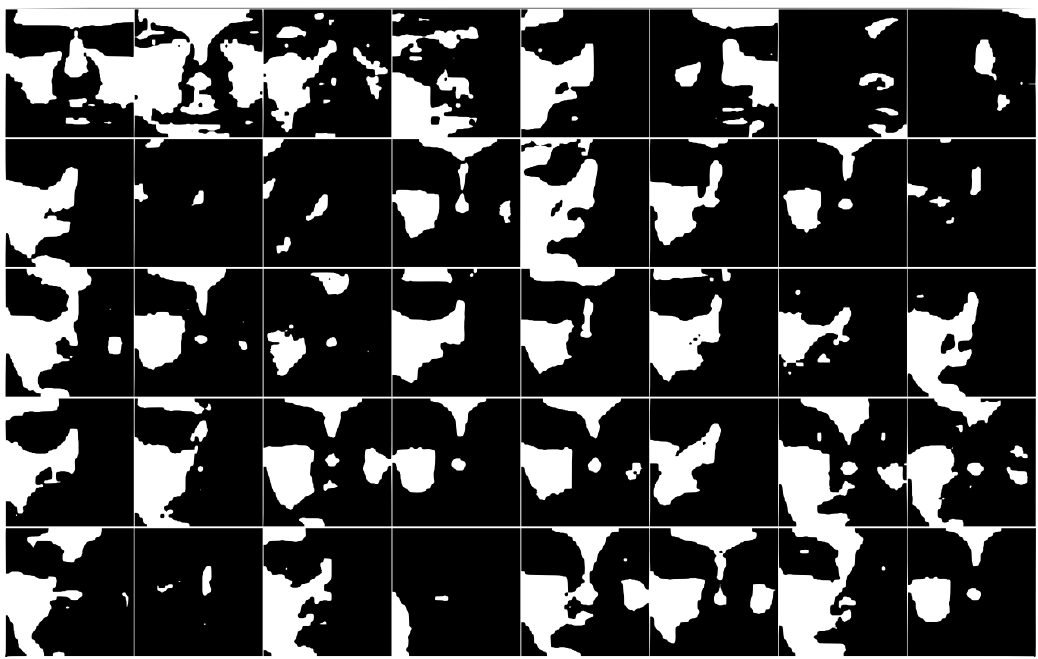}
\caption{Demonstration of the synthesizing process of NG. The first row at the top is some of the original data, and the others are synthesized faces.} 
\label{fig10}
\end{figure}

\bibliography{references.bib}

\end{document}